\documentclass{scrartcl}

\usepackage{amsmath}
\usepackage{amsfonts}
\usepackage{amsthm}
\usepackage{comment}
\usepackage{enumerate}
\usepackage{thmtools}
\usepackage{xcolor}
\usepackage{upgreek}
\usepackage{graphicx}

\newcommand\TombStone{\rule{.7ex}{1.7ex}}
\renewcommand{\qedsymbol}{\TombStone}
\newcommand{\qedd}{\let\qed\relax\quad\raisebox{-.1ex}{$\qedsymbol$}}
\newcommand{\claimqedd}{\let\qed\relax\quad\raisebox{-.1ex}{$\lrcorner$}}

\KOMAoptions{abstract=true}

\usepackage[colorlinks=true]{hyperref}
\hypersetup{
	linkcolor = {blue},
	citecolor = {blue}
}
\usepackage[capitalize, noabbrev]{cleveref}
\usepackage{accents}

\newcommand{\Kb}{{\accentset{\leftrightarrow}{K}}}

\newcommand{\KN}{{\accentset{\leftrightarrow}{K}}_n}

\newcommand{\GN}{{\accentset{\rightarrow}{G}}}

\newcommand{\KD}{{\accentset{\leftrightarrow}{K}}_d}

\newcommand{\KM}{{\accentset{\leftrightarrow}{K}}_{n'}}

\newcommand{\RA}[1]{\mathrel{\stackrel{\makebox[0pt]{{\tiny $#1$}}}{\rightarrow}}}

\newcommand{\IP}{\mathbf{1}}

\newcommand{\RP}{\mathbf{\bar{1}}}

\DeclareMathSymbol{:}{\mathord}{operators}{"3A} %
\newcommand{\TL}[2]{#1 : #2}
\newcommand{\fO}{\mathcal{O}}

\newcommand{\im}[1]{\mathrm{im}(#1)}

\newtheorem{theorem}{Theorem}[section]
\newtheorem{lemma}[theorem]{Lemma}
\newtheorem{corollary}[theorem]{Corollary}

\title{Fixed-point cycles and EFX allocations}

\def\lknote{Institut f\"ur Informatik, Freie Universit\"at Berlin, Germany. E-mail:~\texttt{laszlo.kozma@fu-berlin.de}. {Research supported by DFG grant KO 6140/1-1.}}
\author{Benjamin Aram Berendsohn\thanks{Institut f\"ur Informatik, Freie Universit\"at Berlin, Germany. E-mail:~\texttt{beab@zedat.fu-berlin.de}. {Research supported by DFG grant KO 6140/1-1.}} \and Simona Boyadzhiyska\thanks{Institut f\"{u}r Mathematik, Freie Universit\"{a}t Berlin, Germany. E-mail:~\texttt{s.boyadzhiyska@fu-berlin.de}. {Research supported by the Deutsche Forschungsgemeinschaft (DFG) Graduiertenkolleg ``Facets of Complexity'' (GRK 2434).} } \and L\'{a}szl\'{o} Kozma\thanks{\lknote}}
\date{January 21, 2022}

\begin{document}
		\maketitle

	\begin{abstract}

	We study edge-labelings of the complete bidirected graph $\KN$ with \emph{functions} from the set $[d] = \{1, \dots, d\}$ to itself. We call a cycle in $\KN$ a \emph{fixed-point cycle} if composing the labels of its edges results in a map that has a fixed point, and we say that a labeling is \emph{fixed-point-free} if no fixed-point cycle exists. For a given $d$, we ask for the largest value of $n$, denoted $R_f(d)$, for which there exists a fixed-point-free labeling of $\KN$. Determining $R_f(d)$ for all $d >0$ is a natural Ramsey-type question, generalizing some well-studied zero-sum problems in extremal combinatorics. The problem was recently introduced by Chaudhury, Garg, Mehlhorn, Mehta, and Misra, who proved that $d \leq R_f(d) \leq d^4+d$ and showed that the problem has close connections to \emph{EFX allocations}, a central problem of fair allocation in social choice theory.

In this paper we show the improved bound $R_f(d) \leq d^{2 + o(1)}$, yielding an efficient ${{(1-\varepsilon)}}$-EFX allocation with $n$ agents and $O(n^{0.67})$ unallocated goods for any constant $\varepsilon \in (0,1/2]$; this improves the bound of $O(n^{0.8})$ of Chaudhury, Garg, Mehlhorn, Mehta, and Misra.

{Additionally, we prove the stronger upper bound $2d-2$, in the case where all edge-labels are \emph{permulations}. A very special case of this problem, that of finding \emph{zero-sum cycles} in digraphs whose edges are labeled with elements of $\mathbb{Z}_d$, was recently considered by Alon and Krivelevich and by M\'{e}sz\'{a}ros and Steiner. %
Our result improves the bounds obtained by these authors and extends them to labelings from an arbitrary (not necessarily commutative) group, while also simplifying the proof.}

\end{abstract}

\section{Introduction}
\label{sec:introduction}
Let $\KN$ denote the complete bidirected graph on $n$ vertices, that is, $\KN$ has directed edges $(u,v)$ and $(v,u)$ for every pair $u,v$ of distinct vertices. For an integer $d>0$, a \emph{$d$-labeling}, or simply \emph{labeling} $\ell$ of $\KN$ is an assignment of a function $\ell_e \colon [d] \rightarrow [d]$ to each edge $e$ of $\KN$, where $[d] = \{1, \dots, d\}$.

Let $C$ be a simple cycle of $\KN$ with edges $e_1, \dots, e_k$ (appearing in this order), where each edge $e_i$ is labeled with a function $f_i$, for $i \in [k]$. We say that $C$ is a \emph{fixed-point cycle} if the function $f=f_1 \circ f_2 \circ \cdots \circ f_k = f_k(f_{k-1}(\cdots f_1(x)\cdots))$ has a fixed point, i.e.,\ $f(x) = x$ for some $x \in [d]$. Observe that cyclically permuting the edges along $C$ does not affect the existence of a fixed point.

Let $R_f(d)$ be the largest value $n$ such that a $d$-labeling of $\KN$ with no fixed-point cycle exists. Although not obvious a priori, the finiteness of $R_f(d)$ for all $d$ can be seen through a simple reduction to Ramsey's theorem. The bound obtained in this way, however, is doubly exponential in $d$~\cite{CGMMM21}.

A lower bound of $R_f(d) \geq d$ can be obtained through the following construction (found independently by various authors in different contexts). Say $V(\KD) = [d]$, and label all edges $(i,j)$ such that $i<j$ with the function $x\mapsto x$ and all other edges with the function $x\mapsto x+1 ~(\mbox{mod}~d)$. The avoidance of fixed-point cycles follows from the observation that every simple cycle contains at least one and at most $d-1$ edges of the form $(i,j)$ with $i>j$. 

The question of determining $R_f(d)$ was recently raised by Chaudhury, Garg, Mehlhorn, Mehta, and Misra~\cite{CGMMM21} in a slightly different, but equivalent form they call the \emph{rainbow cycle problem}. They prove the upper bound $R_f(d) \leq d^4 + d$.

The motivation in~\cite{CGMMM21} for introducing this problem comes from an application to \emph{discrete fair division}. In particular, through an elegant and surprising connection, they show that \emph{polynomial} upper bounds on $R_f(d)$ yield nontrivial guarantees for the quality of allocations in a certain natural setting with a strong fairness condition.

\paragraph{EFX allocations.} In economics and computational social choice theory, the \emph{discrete fair division} problem asks to distribute a set of $m$ indivisible goods among $n$ agents in a \emph{fair} way. The problem has a long history (see for example~\cite{Steinhaus}) and different notions of fairness have been extensively studied, leading to a rich set of algorithmic and hardness results (e.g., see~\cite{book1} or~\cite{Caragiannis} and the references therein for an overview and precise definitions).

As simple examples show, complete \emph{envy-freeness} cannot, in general, be achieved. One of the most compelling relaxations of this notion is \emph{envy-freeness up to any good} (EFX). Informally, EFX means that no agent should prefer the goods received by any other agent to their own, if an arbitrary single good is removed from the other's set. The existence of EFX-allocations is considered one of the central open questions of contemporary social choice theory (e.g., see~\cite{Caragiannis, Procaccia, CGMMM21}), and thus, various relaxations of it have been proposed. In particular, it is desirable to show that an EFX-allocation exists (and can be efficiently computed) when (i) a certain global number $t$ of goods are left unallocated, and (ii) envy-freeness is required to hold when every agent scales the values of others' goods by a factor of $1-\varepsilon$ for some $0 \leq \varepsilon < 1$.
Such an allocation is called \emph{$(1-\varepsilon)$-EFX with $t$ unallocated goods}. It is desirable to minimize both $t$ and $\varepsilon$, with the original EFX question requiring $t=\varepsilon=0$.

Very recently, Chaudhury, Garg, Mehlhorn, Mehta, and Misra~\cite{CGMMM21} introduced a correspondence between approximate EFX-allocations and the fixed-point cycle problem described above (in their terminology, the \emph{rainbow cycle problem}). Omitting precise constant factors, the connection can be summarized as follows.

\begin{theorem}[Theorem 4 in~\cite{CGMMM21}]
For any constant $\varepsilon \in (0,1/2]$, if $R_f(d) \in O(d^c)$ for some $c \geq 1$, then there exists a $(1-\varepsilon)$-EFX allocation with $O((n/\varepsilon)^{\frac{c}{1+c}})$ unallocated goods, where $n$ is the number of agents.
\label{thmc}
\end{theorem}

Moreover, if the upper bound on $R_f(d)$ is \emph{constructive} (i.e., a fixed-point cycle can be found in time polynomial in the number $n$ of vertices whenever $n$ exceeds the given bound on $R_f(d)$), then the claimed allocation can also be found in polynomial time. %
The result of $R_f(d) \leq d^4+d$ from~\cite{CGMMM21} is constructive, together with Theorem~\ref{thmc} it thus implies an efficiently computable approximate-EFX %
allocation with $O({n}^{0.8})$ unallocated goods, the first guarantee of this kind with a \emph{sublinear} (in $n$) number of %
unallocated goods. %

\paragraph{Rainbow cycles.}
As remarked in~\cite{CGMMM21}, the question of determining $R_f(d)$ is in itself a natural question of extremal graph theory, independently of its application to EFX-allocations. 

Chaudhury, Garg, Mehlhorn, Mehta, and Misra formulate the problem in a slightly different, but essentially equivalent, way as follows. Given an $n$-partite digraph with each part having \emph{at most} $d$ vertices, a \emph{rainbow cycle} is a cycle that visits each part \emph{at most once}. The task is to determine the largest $n = n(d)$ for which such a digraph can avoid rainbow cycles, with the requirement that each vertex in part $i$ has an incoming edge from at least one vertex in part $j$, for all distinct $i, j\in [n]$. 

Note that we may assume that there is an extremal example containing \emph{exactly} $d$ vertices in each part. Indeed, if $\GN$ is an extremal example and some part has fewer than $d$ vertices, let $v$ be any vertex from that part; then add the necessary number of new vertices, making them all ``clones'' of $v$, that is, every new vertex has the same in- and out-neighborhood as $v$. Then the newly obtained digraph has a rainbow cycle if and only if $\GN$ does. Further, we may assume that each vertex $v$ of $\GN$ has \emph{exactly} one incoming edge from every part other than its own.
With these observations, the equivalence between the rainbow cycle problem and the fixed-point cycle problem is evident: we can view the bipartite digraph containing the edges from a part $V_i$ to another part $V_j$ as the mapping given by $y\mapsto x$, where $(x,y)\in E(V_i, V_j)$.

Our main result improves the upper bound of Chaudhury, Garg, Mehlhorn, Mehta, and Misra.
\begin{restatable}{theorem}{restateThmM}\label{thmm}
For all $d > 0$ we have $R_f(d) \leq d^{2 + o(1)}$. %
\end{restatable}

Similarly to the previous result, our bound is constructive. Combining Theorem~\ref{thmm} with Theorem~\ref{thmc} thus yields the following result about EFX allocations.

\begin{corollary}
For all $\varepsilon \in (0,1/2]$, there exists an efficient $(1-\varepsilon)$-EFX allocation with $O((n/\varepsilon)^{0.67})$ unallocated goods.%
\end{corollary}

\paragraph{Zero-sum problems in extremal combinatorics.}

The problem of determining $R_f(d)$ can be cast as a generalization of some classical zero-sum problems in extremal combinatorics (this is somewhat less apparent in the original multi-partite formulation of the problem). Zero-sum problems have received substantial attention and form a well-defined subfield of combinatorics, with an algebraic flavour. Perhaps the earliest result in this area is the Erd\H{o}s-Ginzburg-Ziv theorem~\cite{EGZ} which states that every collection of $2m-1$ integers contains $m$ integers whose sum modulo $m$ is zero (see~\cite{AlonDubiner} for multiple proofs and extensions). 
Zero-sum problems \emph{in graphs} typically ask, given an edge- or vertex-weighted graph, whether a certain substructure exists with zero total weight (modulo some fixed integer). Well-studied cases include complete graphs, cycles, stars, and trees (e.g., see~\cite{Caro, AlonCaro, AlonLinial, Furedi, Bialostocki, Schrijver} for surveys and representative results).

\smallskip
More recently, for a given positive integer $d$, Alon and Krivelevich~\cite{AK20} asked for the maximum integer $n$ such that the edges of the complete bidirected graph $\KN$ can be labeled with integers so that there is no zero-sum cycle modulo $d$.
Denoting this quantity by $n = R_i(d)$, they showed through an elegant probabilistic argument that $R_i(d) \in O(d\log{d})$, with an improvement to $R_i(d) \in O(d)$ when $d$ is prime. The application considered in~\cite{AK20} is finding cycles of length divisible by $d$ in minors of complete graphs.
It is easy to see this question as a special case of our fixed-point cycle problem: simply replace every edge-label $k$ by the function $x \mapsto x+k ~(\mbox{mod}~d)$. Zero-sum-cycles in the original labeling are then in bijection with fixed-point cycles in our new labeling, and thus $R_i(d) \leq R_f(d)$.

Recently, M\'{e}sz\'{a}ros and Steiner~\cite{MS21} improved the result of Alon and Krivelevich, showing $R_i(d) \leq 8d-1$, with further improvements for prime $d$. In fact, M\'{e}sz\'{a}ros and Steiner generalized the result, allowing the labels to come from an arbitrary commutative group of order $d$. 
The proof of the main result in~\cite{MS21} can be seen as an extension of an incremental construction of~\cite{AK20}, combined with an intricate inductive argument that makes use of group-theoretic results.

We improve these results and show an upper bound of $R_i(d) \leq 2d-2$ through somewhat similar, but arguably simpler arguments. Our result extends to arbitrary groups (not necessarily commutative). In fact, we prove the result in the more general setting of fixed-point cycles, when the edge-labels are restricted to \emph{permutations} of $[d]$. The permutation case subsumes the integer case, since the functions $x \mapsto x+k ~(\mbox{mod}~d)$ in the above reduction are permutations. Denoting the corresponding quantity by $R_p(d)$, we thus have $R_i(d) \leq R_p(d) \leq R_f(d)$, and (omitting the easy case $R_i(1) = R_p(1) = R_f(1) = 1$), prove the following.  

\begin{restatable}{theorem}{restateThmP}\label{thmp}
For all $d \ge 2$, we have $R_p(d) \leq 2d-2$.
\end{restatable}

By the above discussion, Theorem~\ref{thmp} implies that $R_i(d) \leq 2d-2$. In fact, the following more general result holds.

\begin{corollary}\label{cor}
Let $\ell$ be a labeling of $\Kb_{2d-1}$ with elements of a (not necessarily abelian) group  $(G,\cdot)$ of order $d$. Then there is a cycle whose labels multiply to the identity $1 \in G$.
\end{corollary}

As a special case, Corollary~\ref{cor} implies that if $n \geq 2d!-1$ and the edges of $\KN$ are labeled with permutations of $[d]$, then there is a cycle whose labels compose to $\IP_d$.

To our knowledge, the question of permutation-labels has not been considered before. We see it as a natural problem of intermediate generality; it facilitates a simple proof for the integer case, allowing to sidestep the group-theoretic tools used in previous proofs.

\paragraph{Open questions.}
Closing the gap between the lower and upper bounds remains an interesting challenge for all three considered quantities ($R_f(d)$, $R_p(d)$, and $R_i(d)$). The lower bound construction with $d$ vertices discussed earlier can be adapted to all three settings, and it remains a plausible conjecture that $R_i(d) = R_p(d) = R_f(d) = d$. This is easily verified~\cite{CGMMM21} for $d \leq 3$ in the case of $R_f(d)$ (and hence for $R_p(d)$), and was verified via SAT-solvers~\cite{MS21} for $d \leq 6$ in the case of $R_i(d)$. Showing $R_f(d) \in O(d)$  would yield, via Theorem~\ref{thmc}, an approximate EFX-allocation with $O(\sqrt{n})$ unallocated goods.

A further natural question is whether other classical zero-sum results from extremal combinatorics can be extended to the fixed-point setting.

\paragraph{Remark.} After finishing our work, we learned that upper bounds on $R_f(d)$ and $R_p(d)$ similar to ours were independently and concurrently obtained by Akrami, Chaudhury, and Mehlhorn~\cite{acm}.

\paragraph{Organization of the paper.} In \S\,\ref{prelim} we introduce some useful terminology. In \S\,\ref{sec3} we prove Theorem~\ref{thmp}, and in \S\,\ref{sec4} we prove Theorem~\ref{thmm}.

\section{Preliminaries}\label{prelim}

Let $\KN$ be given together with a labeling $\ell$ that assigns functions $[d] \rightarrow [d]$ to the edges.  
By \emph{path} or \emph{cycle} we always mean a \emph{simple} path or cycle. %
All \emph{edges} in this paper are directed, and an edge $(u,v)$ is alternatively denoted by $u \rightarrow v$.

When we say that the edge $u \rightarrow v$ \emph{maps} $x$ to $y$, we mean that the function $\ell_{uv}$ assigned to $u \rightarrow v$ maps $x$ to $y$. Similarly, a path (or cycle) $u_1 \rightarrow \cdots \rightarrow u_k$ maps $x$ to $y$ if, given the label $f_i$ on $u_i \rightarrow u_{i+1}$ for each $i \in [k-1]$, we have $(f_1 \circ \cdots \circ f_{k-1})(x) = y$. For consistency, a path with only one vertex is said to map every value to itself. By $\TL{u}{x} \RA{\ell} \TL{v}{y}$ (or simply $\TL{u}{x} \RA{} \TL{v}{y}$ when the labeling $\ell$ is clear from the context) we denote the fact that the edge $u \rightarrow v$ maps $x$ to $y$.

Let $W$ be a sequence of pairs, each pair $\TL{v}{i}$ consisting of a vertex $v \in V(\KN)$ and a value $i \in [d]$, denoted as $W = (\TL{v_1}{i_1}, \TL{v_2}{i_2}, \dots, \TL{v_k}{i_k})$. We say that $W$ is a \emph{valued walk} in $(\KN, \ell)$ if $\TL{v_j}{i_j} \rightarrow \TL{v_{j+1}}{i_{j+1}}$ for each $j \in [k-1]$. If $v_1 = v_k$, then $W$ is a \emph{valued circuit}. If all vertices in $W$ are distinct, then $W$ is a \emph{valued path}. A valued circuit in which all vertices but the last are distinct is a \emph{valued cycle}.

Finally, let $u, v, w \in V(\KN)$ and $x, y \in [d]$. We say that $w$ \emph{routes $\TL{u}{x}$ to $\TL{v}{y}$} if $\TL{u}{x} \rightarrow \TL{w}{z} \rightarrow \TL{v}{y}$ for some $z \in [d]$.

\section{Permutation labels}\label{sec3}

In this section we prove our main result for permutation labels.

\restateThmP*

First, we introduce a useful tool, generalizing a technique used in~\cite{AK20,MS21} to the fixed-point cycle setting. 

Let $\ell$ be a labeling of $\KN$ that assigns to every edge $e$ a permutation $\ell_e \colon [d] \rightarrow [d]$. 
Consider an edge $u \rightarrow v$ of $\KN$ with label $\ell_{uv}$.  %
Let $(g_i)_i$ and $(h_i)_i$ denote the labels of the incoming and outgoing edges at $v$, respectively. A \emph{shifting at} $u \rightarrow v$ changes the labeling $\ell$ by replacing each $g_i$ by $g_i \circ {\ell_{uv}}^{-1}$ and each $h_i$ by $\ell_{uv} \circ h_i$. In particular, the label of $u \rightarrow v$ becomes the identity permutation $\IP_d$. Let $\ell'$ be the resulting labeling. Observe that mappings along cycles remain unchanged, in particular $(\KN,\ell')$ has a fixed-point cycle if and only if $(\KN,\ell)$ does.

\begin{proof}[Proof of Theorem~\ref{thmp}]

We start with the case $d = 2$. Suppose there is a permutation $2$-labeling $\ell$ of $\Kb_3$ without a fixed-point cycle. The two possible labels are $\IP_2$ and the function $\RP_2$ that maps $1 \mapsto 2$ and $2 \mapsto 1$. For each pair of vertices $u, v$ the labels $\ell_{uv}$ and $\ell_{vu}$ must be distinct, otherwise we have a fixed-point cycle $u \rightarrow v \rightarrow u$. This means that the number of edges labeled $\RP_2$ in $\ell$ is precisely $3$. Now consider two directed $3$-cycles in $\Kb_3$ that form a partition of the edges. One of these cycles has to contain an even number of edges labeled $\RP_2$. Clearly, that cycle maps each value to itself, a contradiction. Thus, $R_p(2) \leq 2$. 

\smallskip

\begin{figure}
\centering
\includegraphics[page=1, scale=0.7]{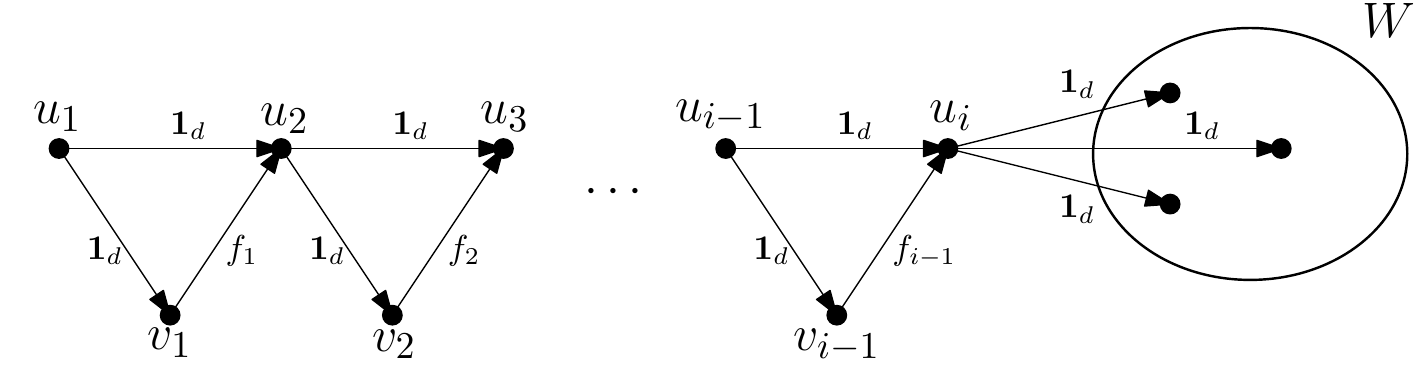}
\caption{Illustration of the proof of Theorem~\ref{thmp}.}\label{fig1}
\end{figure}

Consider now the case $d \ge 3$; let $n \geq 2d - 1$, and let $\ell$ be an arbitrary permutation $d$-labeling of $\KN$. We show that a fixed-point cycle exists.

We construct a ``chain'' consisting of vertices $u_1, \dots, u_j$ and $v_1, \dots, v_{j-1}$ (for some $j \leq d$) and transform the labeling so that each edge of the form $u_i \rightarrow u_{i+1}$ and $u_i \rightarrow v_i$ is labeled $\IP_d$. By \emph{step $i$} we mean either the edge $u_i \rightarrow u_{i+1}$ or the path $u_i \rightarrow v_i \rightarrow u_{i+1}$. Let $S_i \subseteq [d]$ denote the set of possible values to which $1$ can be mapped along some path that concatenates steps $1, \dots, i-1$. Observe that if $S_i = [d]$ then we are done, since $u_{i} \rightarrow u_1$ maps some $x \in [d]$ to $1$, and adding this edge to the path from $u_1$ to $u_{i}$ that maps $1$ to $x$ yields a fixed-point cycle.

 We construct the chain step-by-step (see Figure~\ref{fig1}), ensuring that %
$|S_i| \geq i$ for all $i$. If we reach $|S_i| = d$, then we are done, having used at most $2d-1$ vertices.

  Pick vertex $u_1 \in V(\KN)$ arbitrarily. The condition is then trivially satisfied.

Assume now that we have identified vertices $u_1, \dots, u_i$ and $v_1, \dots, v_{i-1}$ of the chain. %
Let $W \subseteq V(\KN)$ denote the set of vertices \emph{not used} yet, and shift at all edges $u_i \rightarrow u$ for $u \in W$. Observe that no label along the chain is affected. Now consider all edges \emph{between} vertices in $W$. There are two possible cases.

\noindent\underline{\emph{Case 1}}: If some edge $v \rightarrow u$ (for $u,v \in W$) maps some element $x \in S_i$ to some element $y \notin S_i$, then extend the chain with $u_{i+1} = u$ and $v_{i} = v$. The set of reachable values becomes $S_{i+1} \supseteq S_i \cup \{y\}$, establishing the claim for the next step.

\noindent\underline{\emph{Case 2}}: All edges within $W$ map $S_i$ to $S_i$, and consequently $[d] \setminus S_i$ to $[d] \setminus S_i$. Since the chain has used up $2i-1 \leq 2|S_i|-1$ vertices, the digraph induced by $W$ has ${2d-1-(2|S_i|-1)}  \geq 2 \left|[d]\setminus S_i \right|$ vertices. If $i \le d-2$, then $|[d] \setminus S_i| \ge 2$, and we can argue inductively that the digraph induced by $W$ has a fixed-point cycle. If $i = d-1$, then $|[d] \setminus S_i| = 1$ and $|W| \ge 2$, so we trivially have a fixed-point cycle.  \qedd
\end{proof}

We next show the extension of the result to the case of labels from an arbitrary (not necessarily abelian) group. %

\begin{proof}[Proof of Corollary~\ref{cor}]
Let $\ell$ be a labeling that assigns elements of the group $(G,\cdot)$ of order $d$ to edges of $\Kb_{2d-1}$.
Construct a labeling $\ell'$ of $\Kb_{2d-1}$, assigning the function ${x \mapsto x \cdot k}$ to every edge with label $k$ in $\ell$. By Theorem~\ref{thmp}, $\ell'$ has a fixed-point cycle. Suppose its labels are $f_1,\dots, f_t$, where $f_i(x) = x \cdot k_i$ for all $i \in [t]$. Then ${(f_1 \circ \cdots \circ f_t)(x) = x}$ for some $x \in G$, which implies that $k_1 \cdot \,\cdots \, \cdot k_t = 1 \in G$. \qedd
\end{proof}

\section{General function labels}\label{sec4}

Before proving our main theorem, we first present a weaker result, in order to introduce some techniques that will be used later. The result already improves the bound from~\cite{CGMMM21}, while using a similar argument.

\begin{lemma}\label{lemw}
For all $d>0$ we have $R_f(d) \leq d^3 -d^2 + d$.
\end{lemma}
\begin{proof}
Suppose that $n \geq d^3 -d^2 + d + 1$, and let $\ell$ be an arbitrary $d$-labeling of $\KN$. We show that a fixed-point cycle exists. For that, we proceed algorithmically.

Partition $V(\KN)$ arbitrarily into two parts $V = \{v_1, \dots, v_{d}\}$ and $U = \{u_1, \dots, u_{d^3-d^2+1}\}$, and consider all vertices of $U$ initially \emph{unmarked}.

In the preprocessing phase, for each triplet $(i,x,y) \in [d-1] \times [d]^2$ in turn, check whether there exists an unmarked vertex $u \in U$ that routes $\TL{v_i}{x}$ to $\TL{v_{i+1}}{y}$. If yes, then \emph{mark} $u$, and say that $u$ is \emph{responsible} for the triplet $(i,x,y)$. 

Observe that as we mark at most $(d-1)d^2$ vertices of $U$, we have at least one unmarked vertex remaining in $U$ after the preprocessing phase. Let $c \in U$ be such a vertex.  
Consider now the walk that alternates between visiting $c$ and visiting the vertices of $V$ in order, giving rise to the valued walk $W = (\TL{c}{c_0}, \TL{v_1}{x_1}, \TL{c}{c_1}, \TL{v_2}{x_2}, \dots, \TL{v_d}{x_d}, \TL{c}{c_d})$, where $x_i,c_i \in [d]$ for all $i \in [d]$ and $c_0 = 1$. %
Since $|\{c_0,\dots,c_d\}| \leq d$, there must be some $i,j$ such that $0 \leq i<j \leq d$ and $c_i = c_j$, yielding the valued fixed-point \emph{circuit} $C = (\TL{c}{c_i}, \TL{v_{i+1}}{x_{i+1}}, \dots, \TL{v_j}{x_j}, \TL{c}{c_j})$. 

It remains to transform $C$ into a cycle. For all $k$ such that $i< k <j$ replace the subpath $v_k \rightarrow c \rightarrow v_{k+1}$ in $C$ by $v_k \rightarrow c' \rightarrow v_{k+1}$, where $c'\in U$ is the \emph{unique} marked vertex that is responsible for the triplet $(k,x_k,x_{k+1})$. Such a vertex must exist in $U$, for otherwise $c$ itself would have been chosen as responsible for this triplet in the preprocessing phase. We have removed all occurrences of $c$ in $C$ but the first and last, and we have not changed the mappings of values; thus we obtain a fixed-point cycle. An efficient algorithm for finding this cycle is implicit in the proof.  \qedd
\end{proof}

Next, we introduce a transformation that will be useful in the main proof.

\paragraph{Compression.} 
Given a $d$-labeling $\ell$ of $\KN$, define the \emph{imageset} of a vertex $v$ as $\im{v} = \{y ~|~ \TL{u}{x} \RA{\ell} \TL{v}{y}, ~\mbox{for some $u,x$}\}$. In words, $\im{v}$ is the subset of values in $[d]$ that can be mapped to by edges to $v$.

We say that $v$ is \emph{$k$-compressed} if $|\im{v}| \leq k$.  
Observe that all vertices are trivially $d$-compressed, and the existence of a $1$-compressed vertex would immediately yield a fixed-point cycle. Indeed, if $v$ is $1$-compressed, then the cycle $v \rightarrow u \rightarrow v$ maps the unique element in $\im{v}$ to itself for any $u \in V(\KN)$.

We now describe the compression operation, illustrated in Figure~\ref{fig2}. Let $w$ be a $k$-compressed vertex for some $k \geq 2$ and $w' \in V(\KN) \setminus \{w\}$. Suppose there exist two paths $P_1$ and $P_2$ from $w$ to $w'$, together with two distinct values $i_1,i_2\in \im{w}$ and $j\in[d]$ such that $P_1$ maps $i_1$ to $j$ and $P_2$ maps $i_2$ to $j$.  Note that the sets of interior vertices of $P_1$ and $P_2$ do not need to be disjoint and that either of the paths may consist of a single edge.

Define the function $f \colon [d] \rightarrow [d]$ as follows: let
$f(i_2) = j$; for all $x \in \im{w} \setminus \{i_2\}$, let $f(x) = y$, where $P_1$ maps $x$ to $y$; for all $x\in[d]\setminus \im{w}$, choose $f(x)$ arbitrarily.
Now, delete all vertices of $P_1 \cup P_2$, and add a new vertex $w^{\star}$ with edges to and from all remaining vertices.
For all $v \in V(\KN) \setminus {(P_1 \cup P_2)}$, if an edge $v \rightarrow w$ had the label $g$, then the edge $v \rightarrow w^{\star}$ gets the label $g \circ f$, and if an edge $w' \rightarrow u$ had the label $h$, then the edge $w^{\star} \rightarrow u$ gets the label $h$. The labels of edges not involving $w^{\star}$ remain unchanged. 

We refer to this operation as \emph{compressing $P_1$ and $P_2$ to $w^{\star}$}. Suppose we are left with $n'$ vertices and observe that $n-n' \leq |P_1| + |P_2| - 3$, since $P_1$ and $P_2$ have common endpoints. Let $\ell'$ denote the resulting labeling of $\KM$.  We prove two crucial properties. 

\begin{lemma}\label{lem_compression}
Suppose that $P_1$ and $P_2$ (with starting vertex $w$) are compressed to $w^\star$ in the way described above, resulting in a labeling $\ell'$ of $\KM$. 
\begin{enumerate}[(i)]
\item If $w$ is $k$-compressed in $(\KN,\ell)$, then $w^{\star}$ is $(k-1)$-compressed in $(\KM,\ell')$.
\item If $(\KM,\ell')$ has a fixed-point cycle, then $(\KN,\ell)$ has a fixed-point cycle.
\end{enumerate}
\end{lemma}
\begin{proof}\ \\
\vspace{-0.25in}

\begin{enumerate}[(i)]

\item Let $S \subseteq [d]$ denote the set of values to which values in $\im{w} \setminus \{i_2\}$ are mapped by $P_1$. Clearly $|S| \leq |\im{w}|-1 \leq k-1$. Since $P_2$ maps $i_2$ to $j$ and $j \in S$ by construction (as $P_1$ maps $i_1$ to $j$), every edge $v \rightarrow w^{\star}$ maps all values in $[d]$ to $S$ and thus $|\im{w^{\star}}| \leq |\im{w}|-1$.  

\item If a fixed-point cycle in $(\KM,\ell')$ avoids $w^{\star}$, then it also exists in $(\KN,\ell)$. Otherwise, suppose a cycle in $(\KM,\ell')$ contains the segment $\TL{v}{x} \rightarrow \TL{w^{\star}}{y} \rightarrow \TL{u}{z}$. Then, in $(\KN,\ell)$, the edge $v \rightarrow w$ maps $x$ to some value in $\im{w}$ that is mapped by $P_1$ or $P_2$ to $\TL{w'}{y}$, and $w' \rightarrow u$ maps $y$ to $z$. Replacing $w^{\star}$ by $P_1$ or $P_2$ thus gives a fixed-point cycle in $(\KN,\ell)$.
Given a fixed-point cycle in $(\KM,\ell')$, a fixed-point cycle in $(\KN,\ell)$ can be reconstructed (i.e.,\ the compression can be undone) efficiently, with minor bookkeeping. \qedd

\end{enumerate}
\let\qed\relax
\end{proof}

\begin{figure}
\centering
\includegraphics[page=4, scale=0.7]{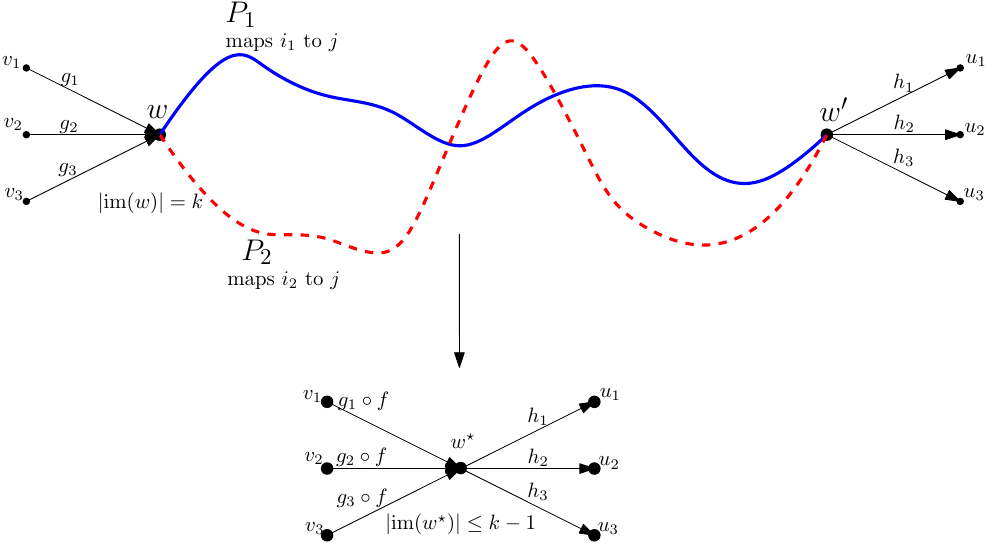}
\caption{Compressing $P_1$ and $P_2$ to $w^{\star}$.}\label{fig2}
\end{figure}

In the remainder of the section we prove our main theorem.

\restateThmM*

The high-level strategy is similar to the one used in the proof of \cref{lemw}. %
The main difference is that, instead of trying to build a cycle using a sequence of $d$ designated vertices ($v_1, \dots, v_d$) and a well-chosen center-vertex ($c$), we pick a sequence of \emph{far fewer} designated vertices, and use the structure imposed upon them by a special compressed vertex. We may fail to find a fixed-point cycle (or even a circuit) with any candidate center. In that case, however, we make progress by \emph{compressing} the special vertex further. After repeating the process sufficiently many times, we will have created a large number of highly compressed vertices. The mappings between these vertices are sufficiently restricted that we can find a fixed-point cycle in the digraph induced by them, through a \emph{recursive application} of the same procedure. 

The following lemma describes the key \emph{win-win} step of the procedure: we either find a fixed-point cycle, or identify two paths that allow the compression of a vertex.

\begin{lemma}\label{lem_main_help}
	Let $k \ge 2$,  $n \ge 4d \lceil \frac{d}{k} \rceil^2 + 2 \lceil \frac{d}{k} \rceil + 2$, and  $w \in V(\KN)$ be a $k$-compressed vertex. Then either $(\KN, \ell)$ has a fixed-point cycle, or there exist a vertex $u \in V(\KN) \setminus \{w\}$, values $i_1, i_2 \in \im{w}$ with $i_1 \neq i_2$ and $x \in [d]$, and two paths on at most $ 4 \lceil \frac{d}{k} \rceil + 2$ vertices each from $\TL{w}{i_1}$ and $\TL{w}{i_2}$ to $\TL{u}{x}$.
\end{lemma}
\begin{proof}

	Let $q = 2 \lceil \frac{d}{k} \rceil + 1$ and arbitrarily fix a $q$-subset $V \subseteq V(\KN) \setminus \{w\}$. Write $V= \{v_1, v_2, \dots, v_q\}$ and $U = V(\KN) \setminus (V \cup \{w\})$. Notice that $|U| \geq  4d \lceil \frac{d}{k} \rceil^2$.
	
	Call a vertex $c \in U$ \emph{valid} for $(i,x,y)$, where $i \in [q-1]$ and $x,y \in [d]$, if $c$ routes $\TL{v_i}{x}$ to $\TL{v_{i+1}}{y}$ and there are at least $q-1$ vertices in $U$ that route $\TL{v_i}{x}$ to $\TL{v_{i+1}}{y}$ (including $c$). Let $k_c$ denote the number of triples $(i,x,y)$ for which $c$ is \emph{not} valid. Double-counting yields
	\begin{align*}
		\sum_{c \in U} k_c \le (q-1) \cdot d^2 \cdot (q-2) < 4 d^2 \left\lceil \frac{d}{k} \right\rceil ^2.
	\end{align*}

	Thus, by the pigeonhole principle, there exists a $c \in U$ with $k_c <  \frac{4d^2\lceil \frac{d}{k} \rceil^2}{|U|} \leq d$. Fix such~a~$c$.
	
Recall that $|\im{w}| = k$ and assume without loss of generality that $\im{w} = [k]$.	For each $i \in [k]$, we construct a (valued) walk $W_i$ that starts with $\TL{v_1}{\ell_{w v_1}(i)}$ and visits $v_2, \dots, v_q$ in that order, possibly taking a detour through $c$ at every step.
	
	Fix $i\in [k]$. We describe the construction of $W_i$ by iteratively constructing prefixes $W_i^j$ that start with $v_1$ and end with $v_j$. First, let $W_i^1 = \TL{v_1}{x_1}$, where $x_1 = \ell_{w v_1}(i)$.

	For $1<j\leq q$, suppose that $W_i^{j-1}$ ends with $\TL{v_{j-1}}{x_{j-1}}$, and consider the values $y_j,x_j$ reached in the valued path $\TL{v_{j-1}}{x_{j-1}} \rightarrow \TL{c}{y_j} \rightarrow \TL{v_j}{x_j}$. %
	If $c$ is valid for $(j-1, x_{j-1}, x_j)$, then let $W_i^j = W_i^{j-1} \rightarrow \TL{c}{y_j} \rightarrow \TL{v_j}{x_j}$. Otherwise, let $W_i^j = W_i^{j-1} \rightarrow \TL{v_j}{x_j'}$ for the appropriate $x_j'$. Finally, let $W_i = W_i^q$. Note that $W_i$ contains at most $2q-1$ vertices (including up to $q-1$ occurrences of $c$).

	Observe that in the process of constructing $W_1, \dots, W_k$, we make $(q-1) \cdot k \ge 2d$ steps in total, where each step either adds $c$ to the current walk, or finds that $c$ is not valid for some triple $(j-1, x_{j-1}, x_j)$.

	We claim that some vertex-value pair occurs at least twice in all the valued walks $W_1, W_2, \dots, W_k$. Suppose not. Then, in particular, $c$ occurs at most $d$ times. Moreover, the triples $(j-1,x_{j-1},x_j)$ for which $c$ is not valid that we encounter in the construction of the walks must be pairwise distinct (if $(j-1,x_{j-1},x_j)$ occurs twice in this way, then two distinct paths must contain $\TL{v_{j-1}}{x_{j-1}}$). This means that we encounter such triples at most $k_c$ times in total. The total number of steps is then at most $d + k_c < 2d$, contradicting our earlier observation. This proves the claim.
	
	%%\smallskip
	
	First, consider the case where some $\TL{v}{x}$ occurs in \emph{two different walks} $W_{i_1}, W_{i_2}$. Define $W_{i_1}'$ and $W_{i_2}'$ by removing all vertices after $v$ from $W_{i_1}$ and  $W_{i_2}$, respectively. We turn $W_{i_1}'$ and $W_{i_2}'$ into (simple) paths as follows. Suppose $c$ occurs more than once in $W_{i_1}'$. Suppose $W_{i_1}'$ contains the segment $\TL{v_i}{x_i} \rightarrow \TL{c}{y_i} \rightarrow \TL{v_{i+1}}{x_{i+1}}$. If this is not the first occurrence of $c$, then we simply replace $c$ by some vertex $c' \in U \setminus \{c\}$ that routes $\TL{v_i}{x_i}$ to $\TL{v_{i+1}}{x_{i+1}}$ and has not been used in this way before. Note that we do this at most $q-2$ times (once for each occurrence of $c$ in $W_{i_1}'$ except for the first one), and there are at least $q-2$ such vertices in $U \setminus \{c\}$ by construction, so the operation is well-defined. We proceed the same way for $W_{i_2}'$ (recall that the paths needed for the compression operation need not be internally disjoint). 
	Call the resulting (simple) paths $W_{i_1}''$ and $W_{i_2}''$, respectively. Then $\TL{w}{i_1} \rightarrow W_{i_1}''$ and $\TL{w}{i_2} \rightarrow W_{i_2}''$ are the desired paths of at most $2q \leq 4\lceil \frac{d}{k}\rceil +2$ vertices, and we are done.
	
	%\smallskip
	Second, suppose that some $\TL{v}{x}$ occurs twice \emph{in a single walk} $W = W_i$. Then $v = c$, since each $v_i$ occurs only once in $W$. We now find a sub-walk of $W$ that has a fixed point. Remove all vertices before the first occurrence of $\TL{c}{x}$ and all vertices after the second occurrence of $\TL{c}{x}$, and call the resulting valued walk $W'$. Clearly, $W'$ is a fixed-point walk and contains at most $q-1$ occurrences of $c$ (counting both start and end individually). Then, similarly as in the first case, transform $W'$ into a simple cycle by replacing all occurrences of $c$, except at the start/end, by suitable vertices (at most $q-3$ of them) in $U \setminus \{c\}$. The result is a fixed-point cycle.  \qedd

\end{proof}

We prove \cref{thmm} via the following recurrence.

\begin{lemma}\label{lem_rec}
	For all $d \ge 2^9$, we have $R_f(d) < 5 d^2 + 9 d \log_2 d \cdot ( R_f(\lfloor \sqrt{d} \rfloor) + 1 )$.
\end{lemma}
\begin{proof}
	Set  $n=5 d^2 + 9 d \log_2 d \cdot ( R_f(\lfloor \sqrt{d} \rfloor) + 1 )$ and consider an arbitrary labeling $\ell$ of $\KN$. We show that $(\KN, \ell)$ contains a fixed-point cycle. We again proceed algorithmically.  Our strategy is to perform transformations on $\KN$ and $\ell$, such as to compress vertices using \cref{lem_main_help}. %
	 We say that a $\lfloor \sqrt{d} \rfloor$-compressed vertex is \emph{fully compressed}. 
	 
	 We first observe that if we have more than $R_f(\lfloor \sqrt{d} \rfloor)$ fully compressed vertices, then we are done. Indeed, consider the subdigraph induced by the set $F$ of vertices that are fully compressed. For all $u,v \in F$, restrict each function $\ell_{uv}$ to  an arbitrary subset of $[d]$ of size $\lfloor\sqrt{d}\rfloor$ containing $\im{u}$. Applying an arbitrary bijection from each such set to $[\lfloor\sqrt{d}\rfloor]$, transform the restricted labeling into a valid $\lfloor \sqrt{d} \rfloor$-labeling on the subdigraph induced by $F$ without creating new fixed-point cycles. Then, if $|F| > R_f(\lfloor \sqrt{d} \rfloor)$, we can recursively find a fixed-point cycle in the induced subdigraph, from which we can recover a fixed-point cycle in~$(\KN, \ell)$.
	
	Now we explain how we obtain the required number of fully compressed vertices. %
	Let $T$ be an arbitrary set of $R_f(\lfloor \sqrt{d} \rfloor)+1$ vertices which are to be compressed and let $S$ be the remaining set of vertices of $\KN$.  At each step, let $w\in T$ be any vertex that is \emph{not} fully compressed; say $w$ is $k$-compressed. Apply \cref{lem_main_help} to $w$ and the subdigraph induced by $\{w\}\cup S$. %
	If we find a fixed-point cycle, then we are immediately done by \cref{lem_compression}. Otherwise, we find two paths $P_1$ and $P_2$ starting at $w$ that we can compress into a single $(k-1)$-compressed vertex $w^\star$.  Now set $T = T\setminus\{w\}\cup\{w^\star\}$ and $S = S\setminus (P_1\cup P_2)$. Note that the size of $S$ reduces by at most $8\lceil\frac{d}{k}\rceil+1 \leq 8 \frac{d}{k}+9$.
	
	We need to ensure that 
	we always have enough vertices to apply \cref{lem_main_help}. First we count the number of vertices that get removed from $S$ throughout the process. Note that, for each $k$ such that $\lfloor \sqrt{d}\rfloor +1 \leq k\leq d$, we transform a $k$-compressed vertex of $T$ into a $(k-1)$-compressed vertex of $T$ at most $R_f(\lfloor \sqrt{d} \rfloor)+1$  times, and every time we perform such an operation, we remove at most $8 \frac{d}{k}+9$ vertices from $S$. Thus, the number of vertices we remove throughout the process is at most
	\begin{align*}
	 (R_f(\lfloor \sqrt{d} \rfloor)+1 )\sum_{k = \lfloor \sqrt{d} \rfloor+1}^d \Big(8 \frac{d}{k} + 9\Big) & ~\le~  (R_f(\lfloor \sqrt{d} \rfloor)+1)\left(8 d \log_2 d + 9d - 1\right)  \\
	 &~\le~ (R_f(\lfloor \sqrt{d} \rfloor)+1)(9 d \log_2 d-1),
	\end{align*}
	where the first inequality uses the fact that $\sum_{k=1}^d 1/k \le \log_2 d$, and the second inequality uses our assumption $\log_2 d \ge 9$.

	Also accounting for the $R_f(\lfloor \sqrt{d} \rfloor)+1$ vertices in $T$, it follows that the final set $S$ has at least $5d^2$ remaining vertices, which is enough to ensure that we can apply \cref{lem_main_help}. Indeed, the number of additional vertices needed to apply \cref{lem_main_help} is
	\begin{align*}
	 4d \Bigl\lceil \frac{d}{k} \Bigr\rceil^2 + 2 \Bigl\lceil \frac{d}{k} \Bigr\rceil + 1 
 ~~\leq ~~& 4d(\sqrt{d}+1)^2 + 2(\sqrt{d}+1)+1 \\
	 ~~\leq ~~& 5 d^2,
	\end{align*}
	since we always have $k \ge \lfloor \sqrt{d} \rfloor + 1 \ge \sqrt{d}$, and $d \geq 16$. 	
	Thus, the process terminates successfully, leaving us with the set $T$ of $R_f(\lfloor \sqrt{d} \rfloor)+1$ fully compressed vertices, using which we find a fixed-point cycle, as discussed above. \qedd
\end{proof}

Finally, we show that the above recurrence  solves to $R_f(d) \in \fO( d^2 \cdot  2^{(\log_2{\log_2{d}})^2} )$, thus implying \cref{thmm}. %

\begin{lemma}\label{lem_calc}
For all $d \geq 4$, we have $R_f(d) \leq 16 \cdot d^2 \cdot 2^{(\log_2{\log_2{d}})^2}$.
\end{lemma}
\begin{proof}

If $d < 2^{20}$, the claim follows through Lemma~\ref{lemw}, since $16 \cdot d^2 \cdot 2^{(\log_2{\log_2{d}})^2} \geq d^3-d^2 + d$ for all $d<2^{20}$. 

Assume therefore that $d\geq2^{20}$.  Denoting $C = 16$ we have:
\begin{align*}
R_f(d) & \leq 5 d^2 + 9 d \log_2 d \cdot ( R_f(\lfloor \sqrt{d} \rfloor) + 1 ) \tag{\small from Lemma~\ref{lem_rec}}\\
& \leq 5 d^2 + 9 d \log_2{d} \cdot \left( C \cdot (\lfloor \sqrt{d} \rfloor)^2 \cdot 2^{(\log_2{\log_2{\lfloor \sqrt{d} \rfloor}})^2} +1 \right) \tag{\small by induction}\\
& \leq 5 d^2 + 9 d \log_2{d} \cdot \left(C \cdot d  \cdot 2^{(\log_2{\log_2{\sqrt{d}}})^2}+1 \right) \tag{\small using $\lfloor \sqrt{d} \rfloor \leq \sqrt{d}$}\\
& \leq 6 d^2 + 9C \cdot d^2 \log_2{d} \cdot 2^{(\log_2{\log_2{\sqrt{d}}})^2} \tag{\small using $9 \log_2{d} \leq d$, for $d \geq 52$}\\
& \leq 10C \cdot d^2 \log_2{d} \cdot 2^{(\log_2{\log_2{\sqrt{d}}})^2} \tag{\small using $d \ge 2^6$}\\
& = 10C \cdot d^2 \log_2{d} \cdot 2^{((\log_2{\log_2{d}}) - 1)^2} \\
& = 10C \cdot d^2 \cdot 2^{((\log_2{\log_2{d}}) - 1)^2+\log_2{\log_2{d}}}\\
& = C \cdot d^2 \cdot 2^{(\log_2{\log_2{d}})^2} \cdot \left( 10 \cdot 2^{1- \log_2{\log_2{d}}} \right)\\
& = C \cdot d^2 \cdot 2^{(\log_2{\log_2{d}})^2} \cdot \left( 20/\log_2{d} \right)\\
& \leq C \cdot d^2 \cdot 2^{(\log_2{\log_2{d}})^2}.\tag{\small using $d \geq 2^{20}$} \\
\end{align*}%
This concludes the proof. \qedd
\end{proof}

Note that the constant factor $16$ in Lemma~\ref{lem_calc} can be significantly reduced through a more careful optimization, but we have avoided this, preferring a simpler presentation. 

\newpage
\small
\bibliographystyle{halphashort}
\bibliography{fixpoint_cycles}

\newcommand{\etalchar}[1]{$^{#1}$}
\begin{thebibliography}{CGM{\etalchar{+}}21}
\expandafter\ifx\csname url\endcsname\relax
  \def\url#1{\texttt{#1}}\fi
\expandafter\ifx\csname doi\endcsname\relax
  \def\doi#1{\burlalt{\nolinkurl{doi:#1}}{http://dx.doi.org/#1}}\fi
\expandafter\ifx\csname urlprefix\endcsname\relax\def\urlprefix{URL }\fi
\expandafter\ifx\csname href\endcsname\relax
  \def\href#1#2{#2}\fi
\expandafter\ifx\csname burlalt\endcsname\relax
  \def\burlalt#1#2{\href{#2}{#1}}\fi

\bibitem[AC93]{AlonCaro}
Noga Alon and Yair Caro.
\newblock On three zero-sum {R}amsey-type problems.
\newblock {\em J. Graph Theory}, 17(2):177--192, 1993.
\newblock \doi{10.1002/jgt.3190170207}.

\bibitem[ACM]{acm}
Hannaneh Akrami, Bhaskar~R. Chaudhury, and Kurt Mehlhorn.
\newblock Rainbow cycle number.
\newblock {\em (Forthcoming)}.

\bibitem[AD93]{AlonDubiner}
Noga Alon and Moshe Dubiner.
\newblock Zero-sum sets of prescribed size.
\newblock {\em Combinatorics, Paul Erd{\H{o}}s is eighty}, 1:33--50, 1993.

\bibitem[AK21]{AK20}
Noga Alon and Michael Krivelevich.
\newblock Divisible subdivisions.
\newblock {\em J. Graph Theory}, 98(4):623--629, 2021.
\newblock \doi{10.1002/jgt.22716}.

\bibitem[AL89]{AlonLinial}
Noga Alon and Nathan Linial.
\newblock Cycles of length $0$ modulo $k$ in directed graphs.
\newblock {\em J. Comb. Theory, Ser. {B}}, 47(1):114--119, 1989.
\newblock \doi{10.1016/0095-8956(89)90071-3}.

\bibitem[Bia93]{Bialostocki}
Arie Bialostocki.
\newblock {\em Zero Sum Trees: A Survey of Results and Open Problems}, pages
  19--29.
\newblock Springer Netherlands, Dordrecht, 1993.
\newblock \doi{10.1007/978-94-011-2080-7_2}.

\bibitem[BT96]{book1}
Steven~J. Brams and Alan~D. Taylor.
\newblock {\em Fair division -- from cake-cutting to dispute resolution}.
\newblock Cambridge University Press, 1996.

\bibitem[Car96]{Caro}
Yair Caro.
\newblock Zero-sum problems -- {A} survey.
\newblock {\em Discret. Math.}, 152(1-3):93--113, 1996.
\newblock \doi{10.1016/0012-365X(94)00308-6}.

\bibitem[CGH19]{Caragiannis}
Ioannis Caragiannis, Nick Gravin, and Xin Huang.
\newblock Envy-freeness up to any item with high {N}ash welfare: The virtue of
  donating items.
\newblock In {\em EC'19: The 19th {ACM} Conference on Economics and
  Computation}, pages 527--545. {ACM}, 2019.
\newblock \doi{10.1145/3328526.3329574}.

\bibitem[CGM{\etalchar{+}}21]{CGMMM21}
Bhaskar~R. Chaudhury, Jugal Garg, Kurt Mehlhorn, Ruta Mehta, and Pranabendu
  Misra.
\newblock Improving {EFX} guarantees through rainbow cycle number.
\newblock In {\em {EC} '21: The 22nd {ACM} Conference on Economics and
  Computation}, pages 310--311. {ACM}, 2021.
\newblock \doi{10.1145/3465456.3467605}.

\bibitem[EGZ61]{EGZ}
Paul {Erd\H{o}s}, Abraham {Ginzburg}, and Abraham {Ziv}.
\newblock {A theorem in additive number theory}.
\newblock {Bull. Res. Council Israel 10F, 41-43}, 1961.

\bibitem[FK92]{Furedi}
Zolt{\'{a}}n F{\"{u}}redi and Daniel~J. Kleitman.
\newblock On zero-trees.
\newblock {\em Journal of Graph Theory}, 16(2):107--120, 1992.
\newblock \doi{10.1002/jgt.3190160202}.

\bibitem[MS21]{MS21}
Tam{\'{a}}s M{\'{e}}sz{\'{a}}ros and Raphael Steiner.
\newblock Zero sum cycles in complete digraphs.
\newblock {\em Eur. J. Comb.}, 98:103399, 2021.
\newblock \doi{10.1016/j.ejc.2021.103399}.

\bibitem[Pro20]{Procaccia}
Ariel~D. Procaccia.
\newblock Technical perspective: An answer to fair division's most enigmatic
  question.
\newblock {\em Commun. ACM}, 63(4):118, 2020.
\newblock \doi{10.1145/3382131}.

\bibitem[SS91]{Schrijver}
Alexander Schrijver and Paul~D. Seymour.
\newblock A simpler proof and a generalization of the zero-trees theorem.
\newblock {\em Journal of Combinatorial Theory, Series A}, 58(2):301--305,
  1991.
\newblock \doi{10.1016/0097-3165(91)90063-M}.

\bibitem[Ste48]{Steinhaus}
Hugo Steinhaus.
\newblock The problem of fair division.
\newblock {\em Econometrica: Journal of the Econometric Society}, 16:101--104,
  1948.

\end{thebibliography}

\end{document}